%% file: ms.tex
\documentclass[12pt]{article}

\pdfoutput=1

\usepackage{setspace}
\usepackage{geometry}
\usepackage{hyperref}
\usepackage{enumerate}
\usepackage{paralist}
\usepackage[authoryear]{natbib}
\usepackage{amsmath}
\usepackage{amssymb}
\usepackage{amsfonts}
\usepackage{amsthm}
\usepackage{graphicx}
\usepackage[normalem]{ulem}

\newtheorem {definition}{Definition}

\newtheorem {proposition}{Proposition}
\newtheorem {lemma}{Lemma}

\newtheorem {fact }{Fact}

\theoremstyle{remark}

\theoremstyle{definition}
\newtheorem{example}{Example}

\theoremstyle{plain} 
\newtheorem*{definition*}{Definition}
\newtheorem*{theorem*}{Theorem}
\newtheorem*{proposition*}{Proposition}
\newtheorem*{lemma*}{Lemma}
\newtheorem*{claim*}{Claim}
\newtheorem*{subclaim*}{Subclaim}
\newtheorem*{observation*}{Observation}
\newtheorem*{conjecture*}{Conjecture}
\newtheorem*{fact*}{Fact}
\newtheorem*{assumption*}{Assumption}
\theoremstyle{remark} 
\newtheorem*{remark*}{Remark}
\newtheorem*{example*}{Example}

\newtheoremstyle{mystyle}
  {}
  {}
  {\normalfont}
  {}
  {\bfseries}
  {.}
  { }
  {\thmname{#1}\thmnumber{ #2}\thmnote{ (#3)}}

\theoremstyle{mystyle}

\usepackage{ifthen}
\newboolean{showcomments}

\newcommand{\marcelo}[1]{ \ifthenelse{\boolean{wordcount}}{}{\ifthenelse{\boolean{showcomments}}
{\textcolor{magenta}{(M:  #1)}}{}}{}}
\newcommand{\kirill}[1]{ \ifthenelse{\boolean{wordcount}}{}{\ifthenelse{\boolean{showcomments}}
{\textcolor{magenta}{(K:  #1)}}{}}{}}
\newcommand{\leeat}[1]{ \ifthenelse{\boolean{wordcount}}{}{\ifthenelse{\boolean{showcomments}}
{\textcolor{magenta}{(L:  #1)}}{}}{}}

\newboolean{showauxiliar}
\newcommand{\auxiliar}[1]{  \ifthenelse{\boolean{wordcount}}{}{\ifthenelse{\boolean{showauxiliar}}
{#1}{}}{}}

\newboolean{wordcount}
\newcommand{\optional}[1]{\ifthenelse{\boolean{wordcount}}{}{#1}{}}

\usepackage{cleveref}
\makeatletter
\renewenvironment{proof}[1][\proofname] {\par\pushQED{\qed}\normalfont\topsep6\p@\@plus6\p@\relax\trivlist\item[\hskip\labelsep\bfseries#1\@addpunct{.}]\ignorespaces}{\popQED\endtrivlist\@endpefalse}
\makeatother
\crefname{theorem}{theorem}{theorems}
\newenvironment{delayedproof}[1]
 {\begin{proof}[Proof of \Cref{#1}]}
 {\end{proof}}

\usepackage{titlesec}
\titleformat{\section}
  {\scshape\filcenter}{\thesection.}{1em}{}
\titleformat{\subsection}[runin]
        {\normalfont\bfseries}
        {\thesubsection.}
        {0.5em}
        {}
        [.]
\titleformat{\subsubsection}[runin]
        {\normalfont\bfseries}
        {\thesubsubsection.}
        {0.5em}
        {}
        [.]

\usepackage{tikz}
\usetikzlibrary{calc,positioning,decorations.pathreplacing}

\newcommand\xqed[1]{%
  \leavevmode\unskip\penalty9999 \hbox{}\nobreak\hfill
  \quad\hbox{#1}}
\newcommand\demo{\xqed{$\triangle$}}
\usepackage{mathtools}
\usepackage{blkarray}

\usepackage{blkarray}
\usepackage{mathtools}
\usepackage{colortbl}
\usepackage{bm}

\usepackage{caption}
\usepackage{subcaption}

\usepackage{nicematrix}

\usepackage{xparse}
\usepackage{tikz}
\usetikzlibrary{calc,matrix,patterns,shadings,backgrounds}
\pgfdeclarelayer{myback}
\pgfsetlayers{myback,background,main}

\tikzset{myfillcolor/.style = {draw,fill=#1}}%

\NewDocumentCommand{\highlight}{O{blue!40} m m}{%
\draw[myfillcolor=#1] (#2.north west)rectangle (#3.south east);
}

\NewDocumentCommand{\vshade}{O{blue!40} O{white} m m}{%
\draw[bottom color =#1,top color=#2, draw=none] (#3.north west)rectangle (#4.south east);
}

\NewDocumentCommand{\oshade}{O{blue!40} O{white} m m}{%
\draw[right color =#1,left color=#2, draw=none] (#3.north west)rectangle (#4.south east);
}

\NewDocumentCommand{\inshade}{O{blue!40} O{white} m m}{%
\draw[inner color =#1,outer color=#2, draw=none] (#3.north west)rectangle (#4.south east);
}

\NewDocumentCommand{\fillpattern}{O{north west lines} O{blue!50} m m}{%
\draw[pattern=#1, pattern color=#2, draw=none] (#3.north west)rectangle (#4.south east);
}

\definecolor{princetonorange}{rgb}{1.0, 0.56, 0.0}

\usepackage{titling}
\usepackage{atbegshi}
\usepackage{dsserif}
\DeclareMathAlphabet{\mathpzc}{OT1}{pzc}{m}{it}

\definecolor{orange}{RGB}{238,136,102}
\definecolor{lightblue}{RGB}{119,170,221}
\definecolor{pear}{RGB}{187,204,51}

\newtheoremstyle{special}
    {\topsep}
    {\topsep}
    {\itshape}
    {}
    {\bfseries}
    {}
    {0.5em}
    {{\thmname{#1}\thmnumber{ #2$^{\bm*}\!$.}\thmnote{\ \textmd{(#3)}}}}

\theoremstyle{special}

\hypersetup {colorlinks=true, linkcolor=blue,
citecolor = blue}

\begin{document}

\title{Centralized Matching with Incomplete Information}
\author{
Marcelo Ariel Fernandez\thanks{Department of Economics, Johns Hopkins University, fernandez@jhu.edu}, Kirill Rudov\thanks{Department of Economics, Princeton University, krudov@princeton.edu}, and Leeat Yariv\thanks{Department of Economics, Princeton University, CEPR, and NBER, lyariv@princeton.edu} \thanks{We thank Muriel Niederle for many conversations that inspired this project as well as the Editor Dirk Bergemann, four anonymous reviewers, Yeon-Koo Che, Laura Doval, Federico Echenique, Nicole Immorlica, and Fuhito Kojima for very helpful comments and suggestions. We gratefully acknowledge financial support from the National Science Foundation through grant SES-1629613 and from the William S. Dietrich II Economic Theory Center at Princeton University.}
}
\date{ July 6, 2021}

\maketitle

\begin{abstract}
We study the impacts of incomplete information on centralized one-to-one matching markets. We focus on the commonly used Deferred Acceptance
mechanism \citep{gale_college_1962}. We show that many complete-information results are fragile to a small infusion of uncertainty about others' preferences.

\end{abstract}

\newpage

\section{Introduction}
\subsection{Overview}

This paper demonstrates the impacts of incomplete information on matching clearinghouses' outcomes. The large literature on two-sided matching has been successful in its application to various markets, ranging from the matching of doctors to residency positions to the matching of kids to schools. Most of this literature assumes complete information: participants are perfectly informed of all other participants' preferences, in addition to their own. We introduce uncertainty about others' preferences in a centralized matching setting employing the celebrated Deferred Acceptance algorithm (DA), introduced by \cite{gale_college_1962}. We show that several canonical results break down even when minimal such uncertainty is added. 

Stability has been of central importance in the design of matching clearinghouses. In a variety of markets, studies have documented the persistence of stable centralized clearinghouses, and the collapse of others \citep*{roth_jumping_1994, roth_economist_2002, mckinney_collapse_2005}. DA clearinghouses are prevalent in applications. They implement stable outcomes for reported preferences and are incentive compatible when cores are small.\footnote{With complete information, DA mechanisms exhibit many appealing features relative to other stable mechanisms. In particular, DA mechanisms cannot be improved upon in terms of manipulability, see \cite{van2019deferred}.} Even with incomplete information, there is arguably a hope that generated outcomes would be ex-post (or complete-information) stable, at least with singleton cores. Otherwise, interactions following the clearinghouse's operations would potentially undo matches. For example, in the medical residency labor market, a newly-minted doctor matched to her, say, third choice could call her top two choices and form a blocking pair, even absent complete information, see \cite{roth_two-sided_1992}. Indeed, concern about such incidents was part of the impetus for introducing centralized clearinghouses to begin with. We therefore analyze conditions under which ex-post stable outcomes are to be expected when cores are small.

There are two main messages. First, equilibria implementing the ex-post stable outcomes always exist, and in some classes of incomplete-information economies---e.g., ones with assortative preferences on either side---they are unique (Propositions~\ref{prop_uniqueness}, \ref{prop_uniqueness_dropping}). Second, however, even complete-information versions of such economies are fragile to the addition of minimal uncertainty about others' preferences. Maintaining singleton cores, a vanishing fraction of market participants can be added so that many features of stable outcomes break down in equilibrium and a large fraction of participants is affected (Propositions~\ref{prop_fragility}, \ref{prop_fragility_assort}). 

Our constructions illustrate that many classical results are overturned in the presence of incomplete information. Equilibrium outcomes corresponding to unstable matchings may be desirable for the \textit{receiving} side in a DA clearinghouse, despite the traditional view of the proposing side as advantaged. Furthermore, the set of matched individuals may differ across equilibrium outcomes. Unlike in the complete-information benchmark, equilibrium selection can offer a useful instrument for influencing who gets matched. 
We also show that several frequently-used technical simplifications are invalid when information is incomplete. In particular, best-response strategy sets do not necessarily include truncation strategies.

Taken together, these results suggest the importance of accounting for details pertaining to the information market participants have in centralized matching clearinghouses.

\subsection{Related Literature}

\cite{roth_two-sided_1989} offers a particular example of a market with incomplete information on \textit{both} sides in which there exists no stable mechanism implementing the (complete-information) stable matching for each preference realization. Importantly, in his example, preference profiles that can conceivably be realized are associated with \textit{multiple} stable matchings.\footnote{\cite{roth_truncation_1999} consider a class of economies in which there are always optimal truncation strategies in the firm-proposing DA. \cite{coles2014optimal} identify optimal truncation strategies in a large class of markets. We illustrate cases in which truncation strategies are sub-optimal.}
 
With incomplete information, \cite{ehlers_incomplete_2007} link singleton cores and truthful preference revelation as an ordinal Bayesian Nash equilibrium. We assume a singleton core state-by-state. Consequently, core outcomes can be implemented in equilibrium. We show cases that yield \textit{other} equilibrium outcomes.\footnote{\cite{immorlica_information_2020} examine how the design of school-choice systems affects information acquisition. 
\cite{fernandez2020deferred} shows that DA only supports stable outcomes when agents avoid regret.
}

In the decentralized, cooperative setting, there has been a long-standing quest for a natural stability notion allowing for incomplete information. \cite{liu_stable_2014} and \cite{liu_stability_2020} offer such a notion for matching markets with transfers and one-sided incomplete information, similar to ours. \cite{bikhchandani_stability_2017} suggests such a notion for similar markets that do not allow transfers.\footnote{
\cite*{chakraborty_two-sided_2010} study a school-choice setting with one-sided incomplete information. They suggest defining stability together with the underlying clearinghouse. Several papers allow for incomplete information when modelling decentralized interactions as a non-cooperative dynamic game (see \citealp*{ferdowsian_decentralized_2021} and
references therein).}  

A growing empirical literature estimates preferences in centralized matching markets using constraints implied by stability (see review in \citealp*{chiappori_econometrics_2016} as well as \citealp*{agarwal_empirical_2015} and \citealp*{hsieh_understanding_2012}). Our results provide caution for this approach when preference information may be incomplete. Market features guaranteeing small cores when information is complete (see, e.g., \citealp*{ashlagi2017unbalanced} and references therein), may be insufficient for ensuring unique stable equilibrium outcomes. Furthermore, observed small cores for \textit{reported} preferences do not guarantee their truthfulness. 

\section{The Model}\label{Model}

\subsection{The Matching Economy}

A \textit{matching market} is a triplet $\mathcal{M}=(F,W,U)$
composed of a finite set of firms $F=\{f_i\}_{i\in[m]}$, where $[m] = \{1,2,\ldots,m\}$; a finite set of workers $W=\{w_j\}_{j \in [n]}$; and match utilities $U=\{
u_{ij}^{f}, u_{ij}^{w}\}_{(i,j) \in [m] \times [n]}$. For each pair $(i,j)$, $u_{ij}^{f}$ is firm $f_{i}$'s utility from matching with worker $w_{j}$ and $u_{ij}^{w}$ is worker $w_{j}$'s utility from matching with firm $f_{i}$. Let $u_{i\varnothing }^{f}$ and $u_{\varnothing j}^{w}$ denote the utilities of firm $f_{i}$ and worker $w_{j}$ from remaining unmatched, respectively. Without loss of generality, all utilities from being unmatched are normalized to zero, i.e. $u_{i\varnothing }^{f}=u_{\varnothing j}^{w}=0$ for all $i,j$. We assume all preferences are strict. Throughout, we focus on markets where all worker-firm pairs are mutually acceptable; i.e., $u_{ij}^{f}>u_{i\varnothing }^{f}$ and $u_{ij}^{w}>u_{\varnothing j}^{w}$.\footnote{Our negative results do not rely on all agents being acceptable.}

An \textit{economy} is a quintuple $\mathcal{E}=(F,W,\left\{U(\theta )\right\} _{\theta \in \Theta },\Theta ,\Psi )$, where $\Theta$ is a finite set of states with each $\theta \in \Theta $ corresponding to a different market $\mathcal{M}(\theta)=(F,W,U(\theta ))$ with the same set of firms $F$ and workers $W$, and $\Psi $ is a probability distribution over states. Without loss of generality, we assume $\Psi $ has full support on $\Theta$.

\subsection{The Game and Equilibrium Concept}\label{Game}
We consider a \textit{centralized matching economy game} in which, at the outset, a state $\theta \in \Theta$ is selected according to $\Psi$. All firms are informed of the realized state $\theta$ and their match utilities in that state---their ``types'' are fully revealing. Each worker $w_{j}$ is privately informed only of his preferences, $u_{ij}^{w}(\theta )$ for $i \in [m] \cup \{\varnothing\}$, which constitute his ``type.'' We assume workers' private information is not revealing of the state; i.e. for any $\theta,\theta'\in \Theta $, $u_{ij}\equiv u_{ij}^{w}(\theta)=u_{ij}^{w}(\theta')$. These match utilities induce an ordinal preferences profile $\succ =\{\succ_{f_{i}},\succ_{w_{j}}\}_{(i,j)\in [m] \times [n]}$. In the Online Appendix, we consider more general environments, allowing for two-sided incomplete information or richer private types.

Participants simultaneously submit rank-ordered lists over acceptable partners to a centralized clearinghouse, which generates a matching using the firm-proposing DA algorithm.

Since truthful reporting is weakly dominant for firms \citep{roth_two-sided_1989}, we focus on Bayesian Nash equilibria (hereafter, BNE) with truthful firms.\footnote{Any worker strategy that lists the most preferred firm first is weakly undominated, see  \citet{roth_two-sided_1992}. In our analysis, the distinction between ex-ante and interim weakly undominated strategies has no bite, see details in the Online Appendix.}

In this paper, we illustrate how ``small'' deviations from the complete-information benchmark can alter outcomes dramatically in a heavily utilized clearinghouse. In particular, we introduce uncertainty about others' preferences only on one market side. Why is the uninformed side of the market (the workers) on the algorithm's receiving side? Otherwise, since firms are informed of the state, and truth-telling is weakly dominant for the proposing side, the analysis would reduce to that of complete information. Naturally, in general, information may be incomplete on both market sides and a designer may not know which side is better informed. As we show in the Online Appendix, our insights  on the impacts of incomplete information hold for a general class of stable mechanisms, not just the firm-proposing DA. We return to this point when discussing our results.

In the complete-information benchmark, incentive compatibility issues are associated with multiplicity of stable matchings. In order to isolate the impacts of incomplete information on strategic behavior, we assume that each market in the support of $\Psi$ has a unique stable matching.\footnote{Recent literature has identified various conditions under which large markets entail small cores (see, e.g., \citealp{Immorlica2005} and \citealp{ashlagi2017unbalanced}). Our analysis speaks to such settings.}   
All participants reporting their preferences truthfully constitutes a BNE in weakly undominated strategies. Therefore,

\setcounter{proposition}{-1}
\begin{proposition} 
Any economy admits a BNE in weakly undominated strategies implementing the unique (complete-information) stable outcome in each state.
\end{proposition}

Throughout, we refer to an outcome as ``unstable'' if it is unstable for the \textit{true} preferences.

\subsection{Preference Assumptions}

The literature has not yet identified general necessary and sufficient conditions for a market to exhibit a unique stable matching. Throughout our analysis, we focus on the \textit{Sequential Preference Condition} (hereafter, SPC), first introduced by \cite{eeckhout2000uniqueness}.

Formally, a market $\mathcal{M}'=(F',W',U')$ is a
\textit{sub-market} of the original market $\mathcal{M}=(F,W,U)$ if $F'\subseteq 
F$, $W'\subseteq W$, and $U'$ is induced by $U$ when restricted to $F' \times W'$. We say that $(f,w)$ is a \textit{top-top match} for sub-market $\mathcal{M}'$ if, for firm $f$, worker $w$ is the favorite worker in $\mathcal{M}'$, and vice versa. A market satisfies the \textit{SPC} if, up to relabeling, there exists an ordering of the firms $f_1,f_2,\ldots,f_m \in F$ and an ordering of the workers $w_1,w_2,\ldots,w_n \in W$ such that for any $i \leq \min(m, n)$, $(f_i,w_i)$ is a top-top match for the sub-market induced by $\{f_j,w_j\}_{j \geq i}$. For these orderings, we say a pair $(f_i,w_i)$, $i \leq \min(m, n)$, and its respective agents have \textit{order} $i$.

Any market that satisfies the SPC has a unique stable matching that can be derived by partnering top-top pairs in sequence. The condition generalizes numerous others.\footnote{These include $\alpha$-reducibility \citep{clark_uniqueness_2006}, the co-ranking condition \citep{legros_co-ranking_2010}, the universality condition \citep{holzman_matching_2014}, the aligned preferences condition \citep*{ferdowsian_decentralized_2021}, and oriented preferences \citep{reny_simple_2021}.} In particular, the SPC is satisfied for markets in which firms (workers) share the same ranking of workers (firms), corresponding to firms (workers) having \textit{assortative preferences}. 

A given economy $\mathcal{E}=(\left\{\mathcal{M}(\theta)\right\} _{\theta \in \Theta },\Theta ,\Psi )$ satisfies the \textit{SPC} if each market $\mathcal{M}(\theta)$ in its support satisfies the SPC for possibly state-specific orderings of firms $f_{1|\theta},f_{2|\theta},\ldots,f_{m|\theta}$ and workers $w_{1|\theta},w_{2|\theta},\ldots,w_{n|\theta}$. 

As it turns out, the SPC alone does not guarantee a unique equilibrium outcome (see the Online Appendix). The following restriction, however, will be useful for identifying economies that do.

\begin{definition}
An economy $\mathcal{E}=(\left\{\mathcal{M}(\theta)\right\} _{\theta \in \Theta },\Theta ,\Psi )$ satisfies the \textit{SPC*} if it satisfies the SPC and, for any state $\theta \in \Theta$, and any order $i \leq \min(m, n)$, 
{\setlength{\abovedisplayskip}{3pt}
\setlength{\belowdisplayskip}{3pt}
\begin{equation*}
        \text{if $f \succ_{w_{i|\theta}} f_{i|\theta}$, then for any $\theta'$, there exists $i' < i$ (that may depend on $\theta'$) with $f = f_{i'|\theta'}$.}
\end{equation*}}
\end{definition}
The SPC* implies that if, in some state, a worker of a given order $i$ prefers some firm over his stable partner in this state, then in \textit{any} state, this firm must prefer her stable (state-specific) partner over the corresponding (state-specific) worker of the same order $i$.  
In other words, if some firm is ``unreachable'' for a worker, it is ``unreachable'' for other workers of the same order in any of the possible states.

Our economy game, and the SPC*, do not treat firms and workers symmetrically. In particular, an economy with (possibly state-dependent) assortative preferences for firms may violate the SPC* (see the Online Appendix).\footnote{Any economy with assortative preferences for workers satisfies the SPC*.} In general, the SPC* holds whenever the order of firms induced by the SPC is state-independent.

\section{Motivating Example}
\label{exmp_motivating}

We start by illustrating a simple incomplete-information economy in which unstable outcomes result from equilibria in weakly undominated strategies.

\newpage

\begin{example}

Consider an economy with three firms and three workers, $m=n=3$, and two states of the world, $\Theta=\{1,2\}$, which are equally likely. Figure~\ref{fig_motivating} describes, for state $\theta \in \Theta$, the preferences. In this matrix notation, rows correspond to firms and columns correspond to workers: 
$u_{ij}(\theta)=(u_{ij}^{f}(\theta),u_{ij}^{w})$.  
We assume that remaining unmatched generates a payoff of $0$ for any market participant.

\begin{figure}[htb]
\input{examples/ex_motivating}
    \caption{Economy with four equilibrium outcomes in which firms report truthfully}
    \label{fig_motivating}
\end{figure}

In both states, there is a unique complete-information stable matching highlighted in bold. Despite information being incomplete, agents have no uncertainty regarding the (complete-information) stable matching. The only difference between the two states appears in $f_{2}$'s preferences---she ranks $w_{1}$ and $w_{3}$ differently across the states. In particular, if all participants knew the state to be $\theta$, the resulting unique stable matching and unique equilibrium outcome of the DA mechanism would be the matching $\mu ,$ $\mu (f_{i})=w_{i},$ $i \in \{1,2,3\}$.

As stated, with firms truthfully reporting, the complete-information stable matching $\mu$ in each state is an equilibrium outcome of this game. We now show that it is not unique: there are three other unstable equilibrium outcomes, all supported by weakly undominated strategies.  The resulting matchings in each state are depicted in Figure~\ref{fig_motivating}.

Consider the $\lambda_1$ equilibrium outcome. In the corresponding equilibrium, $w_{2}$ reports his preferences truthfully, while $w_{1}$ and $w_{3} $ \textit{drop} $f_{1}$ from their preference list, declaring her unacceptable. In particular, truncation strategies are not best responses for both $w_{1}$ and $w_{3}$, unlike the complete-information setting, when there are always truncation best responses (see \citealp*{roth_incentives_1991}).

This equilibrium is appealing for workers---they uniformly prefer it to the one generating the (complete-information) stable outcome in each state. Thus, workers prefer being on the receiving side of DA with this equilibrium selected, see our discussion in Section \ref{Game}. 

What allows for this equilibrium to emerge? In state $1$, were $w_{3}$ and $f_{3}$ absent, the resulting sub-market would have two stable matchings: one matching $w_{i}$ with $f_{i}$; the other, more preferable to the workers, matching $w_{i}$ with $f_{j}$, $i=1,2, j=3-i$. This multiplicity is generated by a \textit{cycle} in the sub-market: $f_{1}$ prefers $w_{1}$, who prefers $f_2$, who prefers $w_{2}$, who prefers $f_{1}$. Worker $w_{1}$ would then benefit from ``truncating'' his preferences and reporting $f_{1}$ as unacceptable. Of course, $w_{3}$ is not absent and could attempt blocking such a matching with $f_{1}$. Why does he drop his claim for $f_{1}$ as well? Consider state $2$. In the sub-market without $w_{1}$ and $f_{3}$, there is again a preference cycle yielding two stable matchings. Worker $w_{3}$ ``truncating'' and reporting $f_{1}$ as unacceptable then guarantees the workers' preferred stable matching. Of course, $w_{1}$ is present and could block this matching with $f_{1}$, but state $1$ ensures that he does not.

Such inter-linked cycles may be generated through the underlying preferences, or through the reported preference profiles. For example, in the $\lambda_{2}$ equilibrium, under the reported preferences, there is a cycle involving all firms and workers in state $1$, although there is no such cycle in the true preferences. Furthermore, the $\lambda_{2}$ and $\lambda_{3}$ equilibrium outcomes are supported only by profiles that entail \textit{permuting} firms in reported preferences. Importantly, the construction of these equilibria is not knife-edge: a small perturbation in match payoffs would not alter their description.

With complete information, when firms use their weakly dominant truth-telling strategy, the set of equilibrium outcomes of the firm-proposing DA coincides with the set of stable matchings and therefore inherits its lattice structure. As this example illustrates, with incomplete information, side-optimality and the existence of ``extremal'' equilibria break down. Indeed, $f_{1}$ and $w_{3}$ prefer the outcome $\lambda_2$ to $\lambda_1$, while $f_{3}$ and $w_{2}$ prefer $\lambda_1$ to $\lambda_2$. 

Furthermore, in this example, the set of unmatched individuals varies across equilibria: $\lambda_3$ is the only equilibrium in which worker $w_3$ is unmatched in state $1$. This stands in contrast to the complete-information setting, where the Rural Hospital Theorem \citep{mcvitie_stable_1970} implies that the set of unmatched individuals is constant across stable matchings, and therefore DA equilibrium outcomes.  

The equilibria all induce unique stable matchings for \textit{reported} preferences in both states. Nonetheless, while corresponding cores are small, reports are not truthful.

The economy in this example does not satisfy the SPC* condition. However, without $f_2$, the resulting ``sub-economy'' effectively entails complete information, satisfies the SPC, and yields a unique equilibrium outcome. As we show, this fragility to the infusion of incomplete information through a small number of participants---here, one---is quite general.\demo 

\end{example}

\section{Stable Equilibrium Outcomes}

In this section, we identify an important class of economies in which incomplete information does not hinder stability and results from the complete-information setting carry over. Nonetheless, in the next section, we show these conclusions are arguably fragile.\footnote{For presentation simplicity, we maintain our assumption of one-sided incomplete information throughout. In the Online Appendix, however, we illustrate that this section's results hold more generally, even when information is incomplete for both market sides.}

\begin{proposition}
\label{prop_uniqueness}
Consider any economy $\mathcal{E}$ such that either 
\begin{inparaenum}[(1)]
\item \label{prop_uniqueness_part1}
firms have (possibly state-specific) assortative preferences or 
\item \label{prop_uniqueness_part2}
the SPC* is satisfied.
\end{inparaenum}
Then, it admits a unique BNE outcome corresponding to the unique (complete-information) stable outcome in each state.
\end{proposition}

Intuitively, consider first the case in which firms have assortative preferences. In each state, firms share the same rankings over workers that they report truthfully. Then, for any state, in the course of the firm-proposing DA, all active firms apply to the same worker. Through the DA steps, the sets of active firms are nested. Therefore, regardless of workers' reports, each worker makes only one choice in the course of the firm-proposing DA and these choices are decisive. In particular, misreporting cannot be beneficial to workers. The only potential impact of a worker misreporting is that he may forgo a firm he would otherwise match to and, consequently, be matched to an inferior firm, or no firm at all. 

The intuition underlying the proposition's second part is more subtle. For any state $\theta$, the corresponding top-top match pair ($w_{1|\theta}$, $f_{1|\theta}$) must be matched under any BNE (Lemma~\ref{lem_top_match} in the Appendix). If not, since firm $f_{1|\theta}$ reports truthfully, worker $w_{1|\theta}$'s top reported firm is not his most preferred firm $f_{1|\theta}$, so that $w_{1|\theta}$ could profitably deviate by shifting $f_{1|\theta}$ to the top of his reported ranking. 

Hence, if there exists an unstable BNE outcome, we can find a state together with an unmatched ``sequential top-top pair'' in this state of order at least two. Let ($w_{i|\theta}$, $f_{i|\theta}$) be an unmatched pair of the smallest possible order $i\geq 2$ among all such pairs and $\theta$ be the corresponding state.\footnote{If there are multiple such pairs, we can pick any.} Again, because firm $f_{i|\theta}$ reports truthfully, worker $w_{i|\theta}$ must be reporting his less desirable equilibrium partner as preferable to $f_{i|\theta}$. We can then construct a deviation strategy such that 
$w_{i|\theta}$ reports truthfully his preferences over firms $\{f:f\succeq_{w_{i|\theta}} f_{i|\theta}\}$ and
maintains the equilibrium ranking for all other firms. 

In any state $\theta$ in which worker $w_{i|\theta}$ is matched with firm $f \preceq_{w_{i|\theta}} f_{i|\theta}$ under the BNE, the deviation ensures that the worker either receives $f$ or a preferable firm from $\{f:f\succeq_{w_{i|\theta}} f_{i|\theta}\}$. 

Suppose that, for some state $\theta' \neq \theta$, the worker is supposed to match with firm $f \succ_{w_{i|\theta}} f_{i|\theta}$ in equilibrium. Both the SPC* and the minimality of $i$ are crucial for establishing that the proposed deviation cannot hurt worker $w_{i|\theta}$. 
The SPC* implies that in state $\theta'$, the ``demanded'' firm $f$ belongs to a sequential pair of order at most $(i-1)$. By the minimality of $i$, in state $\theta'$, the corresponding sequential pair $(w_{i|\theta}, f)$ of order at most $(i-1)$ must be matched. The proposed deviation must then deliver a firm weakly preferable to $f$ under the true preferences. Otherwise, the recursive use of minimality of $i$ and stability for reported preferences yields a contradiction. Intuitively, since $i$ is minimal, other workers' reports do not impede on worker $w_{i|\theta}$ getting his stable match.

In general, equilibria yielding unstable outcomes may exploit cycles in the \textit{true preferences}, as in the motivating example's first equilibrium, or generate cycles through equilibrium reports involving permutations, as in the motivating example's second equilibrium.\footnote{See also Examples 7 and 8 in the Online Appendix, which illustrate unstable equilibrium outcomes in two-state economies in which one state exhibits no (true) preference cycles.}

Do we restore stability if we prohibit permutations in reported preferences and focus on economies with no cycles in true preferences? The result below answers affirmatively.\footnote{A dropping strategy for a worker declares some firms unacceptable, but ranks others truthfully.}

\begin{proposition}
\label{prop_uniqueness_dropping}
Restrict workers' strategy sets to dropping strategies and consider any economy $\mathcal{E}$ such that either 
\begin{inparaenum}[(1)]
\item \label{prop_uniqueness_dropping_part1}
all markets in its support do not have preference cycles\footnote{This condition is formally called the \textit{aligned preferences condition} \citep*{ferdowsian_decentralized_2021}. It is stronger than the SPC, weaker than requiring assortative preferences for either side, and different from the SPC*. Examples 3-5 in the Online Appendix illustrate connections between these conditions.} or
\item \label{prop_uniqueness_dropping_part2}
the SPC* is satisfied. 
\end{inparaenum}
Then, it admits a unique BNE outcome corresponding to the unique (complete-information) stable outcome in each state.
\end{proposition}
The proof of this result is in the Online Appendix.\footnote{Example 6 there illustrates that one cannot relax either condition to the SPC in both of this section's propositions.}

Common values in matching markets naturally translate to assortative preferences, where all agents on one market side (or both sides) rank agents on the other market side identically. Propositions \ref{prop_uniqueness} and \ref{prop_uniqueness_dropping} guarantee that in such markets, when common values are exhibited by either side of the market, there is a unique equilibrium outcome. We next show that the addition of ``minimal'' private values, through a small fraction of added participants, can alter market outcomes dramatically.

\section{Fragility Results}

We now show that various properties displayed by equilibrium outcomes of complete-information economies are fragile to the infusion of ``minute" uncertainty about others' preferences.

We first consider a complete-information economy satisfying the SPC that, from our results thus far, is potentially more resistant to the infusion of uncertainty. We refer to it as the \textit{original economy}. For simplicity, we assume the original economy is balanced, with equal numbers of firms and workers, $m=n$.\footnote{Somewhat more involved arguments can be used to extend our analysis to imbalanced economies.}
        
We augment the original economy by adding one state, and several workers and firms, with their preferences and others' preferences over them in each state. Each market in the augmented economy entails a unique stable matching and occurs with positive probability. Importantly, the preferences of workers and firms present in the original economy remain unchanged with respect to one another in both states. Furthermore, the preferences of workers in the augmented economy coincide across the two states, so that no worker is informed of the state.

Formally, an \textit{augmented economy} of the original market $\mathcal{M}=(F,W,U)$ \textit{with firms $F'$ and workers $W'$} is an economy $\mathcal{E}=(F\cup F',W \cup W',\{U(\theta), U(\theta')\},\{\theta,\theta'\},\Psi)$ such that 
\begin{inparaenum}[(1)]
\item
both $U(\theta)$ and $U(\theta')$ coincide with $U$ when restricted to $F \times W$;
\item
each $w \in  W'$ has identical match utilities in both states; and 
\item
$\Psi$ is non-degenerate and both $U(\theta)$ and $U(\theta')$ are each associated with a unique stable matching.
\end{inparaenum}

\begin{proposition}
\label{prop_fragility}
Consider any balanced market 
$\mathcal{M}$ with at least two agents on each side satisfying the SPC. Then, we can construct an augmented economy $\mathcal{E}$ of $\mathcal{M}$ with one additional firm and one additional worker such that:
\begin{enumerate}
    \item there is a BNE that supports unstable outcomes in both states;
    \item any BNE outcome can be supported by weakly undominated strategies that induce a unique stable matching (for reported preferences) in each state;
    \item the set of matched workers varies across equilibrium outcomes. Furthermore, workers disagree on which is preferred.
\end{enumerate}
\end{proposition}

Certainly, one can always augment a complete-information economy satisfying our restrictions by ``appending'' three workers and three firms that view all of the original economy's participants as least desirable and are viewed by them as least desirable as well. These added firms and workers can exhibit the same preferences over one another as in our motivating example. The conclusions of the proposition would then follow for \textit{any} original complete-information economy, even absent the SPC assumption. The main point of the proposition is that instability can emerge in equilibrium even with minimal augmentation.

To illustrate our construction, the details of which appear in the Online Appendix, consider the motivating example. Workers $\{w_1,w_3\}$ and firms $\{f_1,f_3\}$ exhibit the same match utilities over one another in both states. Furthermore, the corresponding sub-market satisfies the SPC, with $(f_1,w_1)$ as the top-top pair. The motivating example can then be thought of as an augmentation of this complete-information economy with  $f_2$ and $w_2$. 

Consider now an arbitrary SPC market and suppose the pairs $(f_{n-1},w_{n-1})$ and $(f_n,w_n)$ have the highest order (i.e., they constitute the last pairs matched in the recursive top-top matches' coupling). When their preferences over one another mimic those of $\{w_1,w_3\}$ and $\{f_1,f_3\}$ in our motivating example, we can replicate that construction by adding a pair similar to $f_2$ and $w_2$. In other cases, we construct examples analogous to our motivating example that allow the generalization.\footnote{With complete information, increasing competition on one side of the market cannot improve the matches of any participant on that market side and cannot harm the matches of participants on the other market side, see \cite{roth_two-sided_1992}. In the Online Appendix, we provide an incomplete-information example in which the addition of only one worker is beneficial for some workers.} 

Our proof's construction is such that at most three workers need to misreport strategically in order to implement any BNE in the augmented economy. The proof is valid for any arbitrary probability of each state in the augmented economy. Furthermore, if worker $w_n$ exhibits (non-justified) envy towards $w_{n-1}$, i.e. $u^w_{n-1,n}>u^w_{n,n}$, we can construct the augmented economy to entail no uncertainty regarding the unique stable matching. In that case, our construction also ensures that all unstable BNE outcomes ex-ante Pareto dominate the stable one for workers. Thus, unstable BNE outcomes can be more appealing for workers on efficiency grounds and require fairly limited strategizing relative to stable outcomes.

Naturally, one may wonder about the impact of such equilibrium multiplicity on overall outcomes of participants. Particularly when markets are large, as in many applications of DA, whether or not a substantial fraction of participants is affected would have important implications. We now show that, indeed, the addition of a small number of market participants can alter outcomes dramatically.

We start with an example of an $n \times n$ market satisfying the SPC, for which the addition of one firm and one worker allows for a BNE that substantially improves \textit{all} workers' outcomes---in each state, the average rank increase for workers is roughly $n/2$.

\begin{example}
\label{exmp_fragility_2}

Consider a market with $n$ firms $\{f_1,\ldots,f_n\}$ and $n$ workers $\{w_1,\ldots,w_n\}$, with the following preferences. For all $i \in [n]$:
\begin{align*}
f_i&: w_i \succ w_{i+1} \succ \ldots \succ w_n \succ w_{i-1} \succ w_{i-2} \succ \ldots \succ w_1;\\
w_i&: f_{i-1} \succ f_{i-2}\succ \ldots \succ  f_1 \succ f_i \succ f_{i+1} \succ \ldots \succ f_n.
\end{align*}
This market exhibits no preference cycles and has a unique stable matching that matches $w_i$ with $f_i$ for all $i \in [n]$.

Consider now an augmentation with firm $f$ and worker $w$ and corresponding to two equally likely states, $\Theta=\{1,2\}$. 

Choose firm $f$'s preferences as state-independent:
\begin{align*}
    f:w \succ w_n \succ w_1 \succ w_2,\,w_3,\,\ldots,\,w_{n-1}. 
\end{align*}
Extend firm $f_{n-1}$'s preference in a state-dependent fashion. In state 1, the firm ranks $w$ between $w_{n-1}$ and $w_n$, while in state 2, she ranks $w$ between $w_n$ and $w_{n-2}$. All other firms retain state-independent preferences and rank $w$ as their least favorite.

As for workers, worker $w$ has the following preferences:
\begin{equation*}
    w:f_{n-1} \succ f \succ f_{n} \succ f_1,\,f_2,\,\ldots,\,f_{n-2}\\
\end{equation*}
and each of the original workers $w_i$ ranks $f$ right above $f_i$ (for $i>1$, below $f_1$).

Set $w$'s utility from $f_{n-1}$ to be sufficiently high so that his expected utility from matching with $f_{n-1}$ in state 1 and $f_n$ in state 2 is higher than his utility from matching with $f$ in both states. Set $w_n$'s utility from $f$ to be sufficiently close to his utility from $f_n$, so that his expected utility from matching with $f_{n-1}$ in state 1 and $f_n$ in state 2 is higher than his utility from matching with $f$ in both states. All unspecified utilities can be set arbitrarily.

In this economy, both states have the same unique stable matching $\mu$ such that, as before, $\mu(f_{i})=w_{i}$ for any $i \in [n]$, and $\mu(f)=w$.

Consider the strategy profile such that workers $w$ and $w_n$ both drop $f$, and everyone else reports truthfully. It generates unstable outcomes $\lambda(1), \lambda(2)\neq \mu$ in the respective states such that 
\begin{inparaenum}[(1)]
\item  $\lambda(1)(w_1)=\lambda(2)(w_1)=f$;
\item for any $i \neq 1,n$,  $\lambda(1)(w_i)=\lambda(2)(w_i)=f_{i-1}$;
\item $\lambda(1)(w)=\lambda(2)(w_n)=f_{n-1}$; and
\item $\lambda(2)(w)=\lambda(1)(w_n)=f_{n}$.
\end{inparaenum}
These outcomes constitute unique stable matchings for the reported preferences.

The candidate profile constitutes a BNE in weakly undominated strategies. Indeed, worker $w$ cannot get his most preferred $f_{n-1}$ in state $2$. In order to get his second-most preferred $f$ in state $2$, $w$ needs to report $f$ as acceptable. However, such a deviation precludes him from getting his most preferred $f_{n-1}$ in state $1$. By construction, it is not profitable.

Similarly, worker $w_n$ cannot get firms $\{f_j\}_{j \in [n-2]}$ in state $1$. In order to get $f$ in state $1$, $w_n$ needs to report $f$ to be more desirable than $f_n$. However, such a deviation precludes him from getting his most preferred $f_{n-1}$ in state $2$ and is thus unprofitable.

Finally, all other workers have no incentives to deviate since the generated matchings are unique stable for the reported preferences.

In the augmented economy, in either state, the average rank difference for workers is approximately $n/2$ when $\lambda(\theta)$ and $\mu$ are compared.\demo
\end{example}

We now restrict attention to particular economies, those in which preferences are assortative on both sides of the market. While these preferences are certainly a polar extreme, they allow us to generalize the example above and illustrate the substantial impacts even a ``small'' infusion of uncertainty may have.

\begin{proposition}
\label{prop_fragility_assort}
Consider any market $\mathcal{M}$ with $n$ firms and $n$ workers having assortative preferences. Let $k \leq n-2$. There is an economy $\mathcal{E}$ augmented with $k$ firms and $k$ workers such that the following holds. There is a BNE in weakly undominated dropping strategies that matches, identically in both states, $(n-k-1)$ original workers to original firms. According to the original preferences, each of these workers' matched firm has a rank lower by $k$ relative to their stable partner.
\end{proposition}

The proposition implies that, for any $\epsilon>0$, when the economy is large enough, we can add a small volume of agents, accounting for an $\epsilon$ fraction of the economy's population, such that the resulting economy exhibits equilibrium multiplicity. In at least one of the ``unstable'' equilibria, a large fraction of the original workers, roughly $(1-\epsilon)$ of them, experiences a substantial improvement in the ranking of their matched partner. In fact, our construction guarantees that \textit{all} workers in the augmented economy get higher expected payoffs in this equilibrium relative to the one implementing the stable outcome in each state.\footnote{As before, the proof also implies that, under the identified BNE, reported preferences are associated with a unique stable matching in each state.} 

As mentioned in Section \ref{Model}, our results rely on the uninformed workers being on the receiving side of the DA mechanism. Why wouldn't a market designer simply use the worker-proposing DA instead and guarantee stable outcomes? Importantly, if the designer uses the worker-proposing DA, a mirror image of our Example 1---and our entire analysis---in which the labels of workers and firms are switched, would replicate the current message of the paper. In fact, one can straightforwardly generate richer economies in which, with some probability, it is the workers who are uninformed and, with some probability, it is the firms who are uninformed, that would yield our negative results for \textit{either} version of DA. In the Online Appendix, we show that \textit{no quantile stable mechanism} is a panacea.\footnote{Quantile stable mechanisms, which include DA mechanisms, guarantee that, for reported preferences, each participant receives a stable partner in the same quantile. See \cite{teo1998geometry} and references that follow.} A modification of our Example 1 demonstrates that any such stable mechanism would yield multiple equilibrium outcomes in weakly undominated strategies. 

\section*{Appendix -- Proofs}

\begin{lemma}
\label{lem_top_match}
Consider any economy $\mathcal{E}$. If firm $f$ and worker $w$ form a top-top match pair in state $\theta$, they must be matched in this state under any BNE.
\end{lemma}
\begin{proof}
Since firms report truthfully, in state $\theta$, firm $f$ reports worker $w$ as her top choice. If they are not matched, then $w$'s top-ranked firm is not his first true choice $f$. Therefore, $w$ could profitably deviate by shifting $f$ to the top of his ranking to strictly benefit in state $\theta$ without losing in other states.
\end{proof}

\begin{delayedproof}{prop_uniqueness}
We already discussed~(\ref{prop_uniqueness_part1}) in the text. Suppose~(\ref{prop_uniqueness_part2}) holds, the SPC* is satisfied.

Consider first any \textit{two-state} economy with sequential top-top pairs $\{(f_{i|\theta},w_{i|\theta})\}_{i \leq \min(m,n)}$ in state $\theta \in \{1,2\}$, defined in accordance with the SPC*. Since these two sequences correspond to unique stable matchings in the respective states, it is sufficient to prove that all top-top pairs must be matched under any BNE.

We apply induction on $k \leq \min(m,n)$, where $k$ is the number of the first pairs $\{(f_{i|\theta},w_{i|\theta})\}_{i \leq k}$ in the sequences above. The assertion holds for $k=1$. Indeed, by Lemma~\ref{lem_top_match}, top-top matches must be matched under any BNE. Suppose that the assertion holds for $k \geq 1$. By symmetry with respect to states, it suffices to prove that $(f_{k+1|1},w_{k+1|1})$ must be matched to each other in state $1$ under any BNE.

Suppose, towards a contradiction, that $(f_{k+1|1},w_{k+1|1})$ are not matched in state $1$ under some equilibrium yielding matching $\lambda(\theta)$ in state $\theta \in \{1,2\}$. Then, both $f_{k+1|1}$ and $w_{k+1|1}$ get less desirable assignments in state $1$ under $\lambda(1)$. Because $f_{k+1|1}$ reports truthfully, this implies that worker $w_{k+1|1}$ reports the less desirable $\lambda(1)(w_{k+1|1})$ as preferable to $f_{k+1|1}$ in his equilibrium ranking $Q$. To reach a contradiction, it suffices to find a profitable deviation.

Construct a deviation $Q'$ from $Q$ by moving firms $\{f: f \succeq_{w_{k+1|1}} f_{k+1|1}\}$ to the top while keeping truthful ranking of these firms and $Q$'s ranking of all remaining firms. Let $\lambda'(\theta)$ be the resulting stable matching in state $\theta \in \{1,2\}$ for the reported preferences. This deviation delivers $\lambda'(1)(w_{k+1|1}) \succeq_{w_{k+1|1}} f_{k+1|1} \succ_{w_{k+1|1}} \lambda(1)(w_{k+1|1})$ to worker $w_{k+1|1}$ in state $1$. 

We now show that $\lambda'(2)(w_{k+1|1}) \succeq_{w_{k+1|1}}  \lambda(2)(w_{k+1|1})$, i.e. worker $w_{k+1|1}$ cannot be worse off in state $2$. 

First, if $f_{k+1|1} \succeq_{w_{k+1|1}} \lambda(2)(w_{k+1|1})$, the proposed deviation $Q'$ does not hurt worker $w_{k+1|1}$ in state $2$, as desired. 

Second, suppose that $\lambda(2)(w_{k+1|1}) \succ_{w_{k+1|1}} f_{k+1|1}$. Then, by the SPC*, for some $l \leq k$, $f_{l|2} = \lambda(2)(w_{k+1|1})$. Therefore, $w_{k+1|1}=w_{l|2}$ by the induction hypothesis.

Towards a contradiction, suppose that $f_{l|2} =\lambda(2)(w_{k+1|1}) \succ_{w_{k+1|1}} \lambda'(2)(w_{k+1|1})$ for truthful reports. By stability of $\lambda'(2)$ for the reported preferences, it must be the case that $w_{{l_1|2}} \coloneqq \lambda'(2)(f_{l|2}) \succ_{f_{l|2}} w_{l|2}=\lambda(2)(f_{l|2})$. Since $w_{l|2}$ is $f_{l|2}$'s favorite among all $\{w_{t|2}\}_{t \geq l}$, by our construction, we must have $l_1<l$. By the induction hypothesis, $\lambda(2)(w_{{l_1|2}})=f_{{l_1|2}}$.

Because $\lambda(2)$ is stable for the equilibrium reported preferences, under those, $f_{{l_1|2}}=\lambda(2)(w_{l_1|2}) \succ_{w_{{l_1|2}}} f_{l|2}=\lambda'(2)(w_{{l_1|2}})$. By stability of $\lambda'(2)$ for reported preferences, it then follows that $w_{l_2|2} \coloneqq \lambda'(2)(f_{{l_1|2}}) \succ_{f_{{l_1|2}}} w_{{l_1|2}}=\lambda(2)(f_{{l_1|2}})$. Again, by our construction, $l_2 < l_1$ and, therefore, by the induction hypothesis, $\lambda(2)(w_{l_2|2})=f_{l_2|2}$.

We can continue recursively until we reach $l_T = 1$ for some $T \geq 1$, so that we have $w_{1|2} = \lambda'(2)(f_{l_{T-1}|2}) \succ_{f_{l_{T-1}|2}} w_{l_{T-1}|2}=\lambda(2)(f_{l_{T-1}|2})$ and $\lambda(2)(w_{1|2})=f_{1|2}$ as above. Since $\lambda(2)$ is stable for the equilibrium reported preferences, under those, $f_{1|2}=\lambda(2)(w_{1|2}) \succ_{w_{1|2}} f_{l_{T-1}|2}=\lambda'(2)(w_{1|2})$. Also, by our construction, $w_{1|2}$ is $f_{1|2}$'s favorite among all workers, and hence $w_{1|2}  \succ_{f_{1|2}} \lambda'(2)(f_{1|2})$. This contradicts the stability of $\lambda'(2)$ under the reported preferences. Therefore, we must have $\lambda'(2)(w_{k+1|1}) \succeq_{w_{k+1|1}} \lambda(2)(w_{k+1|1})$, as desired.

Finally, for economies with more than two states, an identical proof works. Indeed, the same deviation is strictly beneficial in state $1$ and weakly beneficial in any state $\theta \neq 1$.
\end{delayedproof}

\begin{delayedproof}{prop_fragility_assort}
Take any balanced market $\mathcal{M}$ such that, up to relabeling, all firms agree on $w_1 \succ w_2 \succ \ldots \succ w_n$ and all workers agree on $f_1 \succ f_2 \succ \ldots \succ f_n$. For any given non-degenerate distribution $\Psi$ over the two states, $\Theta=\{1,2\}$, consider this market augmented with firms $\{\mathbb{f}_i\}_{i \in [k]}$ and workers $\{\mathbb{w}_i\}_{i \in [k]}$.

For the new firms, choose state-independent preferences:
{\setlength{\abovedisplayskip}{3pt}
\setlength{\belowdisplayskip}{3pt}
\begin{align*}
    \mathbb{f}_1&: \mathbb{w}_1 \succ w_n \succ w_1 \succ \text{rest},\\
    \mathbb{f}_i,\,i\neq 1&:  \mathbb{w}_i \succ w_i \succ  \text{rest}.
\end{align*}}
Extend firm $f_{n-1}$'s preference in a state-dependent way. In state 1, the firm ranks $\mathbb{w}_{1}$ between $w_{n-1}$ and $w_n$ and views other new workers as least desirable. In state 2, she considers all new workers as least preferred, with $\mathbb{w}_{1}$ being the most desirable among them.

Other firms retain state-independent preferences. Each firm $f_i$, $i \leq n-k-1$, ranks all new workers as least favorite. Any firm $f_i$, $n-k \leq i \leq n-2$, ranks $\mathbb{w}_{n-i}$ between $w_{n-1}$ and $w_n$ and views other new workers as least desirable. The remaining firm $f_n$ views all new workers as least preferred, with $\mathbb{w}_{1}$ being the most preferred among them.

As concern workers, added ones have:
{\setlength{\abovedisplayskip}{3pt}
\setlength{\belowdisplayskip}{3pt}
\begin{align*}
    \mathbb{w}_1&: f_{n-1} \succ \mathbb{f}_1 \succ f_n \succ \text{rest},\\
    \mathbb{w}_i,\,i\neq 1&:  f_{n-i} \succ \mathbb{f}_i \succ  \text{rest}.
\end{align*}}
Any worker $w_i$, $i \leq k$, ranks all new firms as most preferred, with  $\mathbb{f}_{i}$ being the most desirable among them. Each worker $w_i$, $ k+1 \leq i \leq n-1$, ranks all new firms right above $f_i$ (and below $f_{i-1}$). Finally, worker $w_n$ ranks $\mathbb{f}_{1}$ between $f_{n-1}$ and $f_n$ and ranks other new firms as least preferred.

Set $\mathbb{w}_1$'s utility from $f_{n-1}$ to be sufficiently high so that his expected utility from matching with $f_{n-1}$ in state 1 and $f_n$ in state 2 is higher than his utility from matching with $\mathbb{f}_1$ in both states. Set $w_n$'s utility from $\mathbb{f}_1$ to be sufficiently close to his utility from $f_n$, so that his expected utility from matching with $f_{n-1}$ in state 1 and $f_n$ in state 2 is higher than his utility from matching with $\mathbb{f}_1$ in both states. All unspecified utilities can be set arbitrarily.

In this economy, both states have the same unique stable matching $\mu$ such that $\mu(f_{i})=w_{i}$ for any $i \in [n]$, and $\mu(\mathbb{f}_{i})=\mathbb{w}_{i}$ for any $i \in [k]$.

Consider the strategy profile in which workers $\mathbb{w}_1$ and $w_n$ both drop $\mathbb{f}_1$, worker $w_{n-1}$ drops $\{f_{n-k},f_{n-k+1},\ldots, f_{n-2}\}$, and everyone else reports truthfully. It generates unstable outcomes $\lambda(1), \lambda(2)\neq \mu$ in the respective 
states such that 
\begin{inparaenum}[(1)]
\item  for any $i \neq 1$, $\lambda(1)(\mathbb{w}_i)=\lambda(2)(\mathbb{w}_i)=f_{n-i}$;
\item for any $i \in [k]$, $\lambda(1)(w_i)=\lambda(2)(w_i)=\mathbb{f}_{i}$;
\item for any $i \in \{k+1,k+2,\ldots,n-1\}$, $\lambda(1)(w_i)=\lambda(2)(w_i)=f_{i-k}$;
\item $\lambda(1)(\mathbb{w}_1)=\lambda(2)(w_n)=f_{n-1}$; and
\item $\lambda(2)(\mathbb{w}_1)=\lambda(1)(w_n)=f_{n}$.
\end{inparaenum}
These outcomes constitute unique stable matchings for the reported preferences.

The candidate profile constitutes a BNE in weakly undominated strategies. Indeed, worker $\mathbb{w}_1$ cannot get his most preferred $f_{n-1}$ in state $2$. In order to get his second most preferred $\mathbb{f}_1$ in state $2$, $w$ needs to report $\mathbb{f}_1$ as acceptable. However, such a deviation precludes him from getting his most preferred $f_{n-1}$ in state $1$. By construction, it is unprofitable.

Similarly, worker $w_n$ cannot get firms $\{f_1,f_2,\ldots,f_{n-2}\}$ in either state. Also, he cannot get firm $f_{n-1}$ in state $1$. In order to get $\mathbb{f}_1$ in state $1$, $w_n$ needs to report $\mathbb{f}_1$ as more desirable than $f_n$. However, such a deviation precludes him from getting his most preferred $f_{n-1}$ among ``achievable'' firms in state $2$ and is thus unprofitable.

Worker $w_{n-1}$ cannot get firms $\{f_1,f_2,\ldots,f_{n-k-2}\}$ in either state. Consequently, he gets his most preferred $f_{n-k-1}$ among ``achievable'' firms and, hence, has no incentives to deviate.

Finally, other workers have no incentives to deviate since the generated matchings are unique stable for the reported preferences. 
\end{delayedproof}

\newpage 

\addcontentsline{toc}{section}{References}
\bibliographystyle{aer-oxford}
\bibliography{ms.bib}

\end{document}

%% file: examples/ex_motivating.tex
\begin{center}
\begin{tikzpicture}
\matrix (m)[matrix of nodes, style={nodes={color = white, rectangle,draw,minimum width=3em, line width=0pt}}, minimum height=2em, row sep=-\pgflinewidth, column sep=-\pgflinewidth]
{
$\bm{\mu}$	\\ 
\color{black}$\bm{\mu}$		\\ 
$\bm{\mu}$	\\
};
\end{tikzpicture}
\quad
\begin{tikzpicture}
\matrix (m)[matrix of nodes, style={nodes={rectangle,draw,minimum width=3em}}, minimum height=2em, row sep=-\pgflinewidth, column sep=-\pgflinewidth]
{
$\bm{3,2}$		& $1,5$		 	& $2,2$		\\ 
$2,5$				& $\bm{3,2}$		& $1,5$		\\ 
$1,1$				& $2,1$			& $\bm{3,1}$	\\
};

\begin{pgfonlayer}{myback}
\node at (0,1.6) {$U(1)$};
\end{pgfonlayer}
\end{tikzpicture}
\quad
\begin{tikzpicture}
\matrix (m)[matrix of nodes, style={nodes={rectangle,draw,minimum width=3em}}, minimum height=2em, row sep=-\pgflinewidth, column sep=-\pgflinewidth]
{
$\bm{3,2}$		& $1,5$		 	& $2,2$		\\ 
$1,5$				& $\bm{3,2}$		& $2,5$		\\ 
$1,1$				& $2,1$			& $\bm{3,1}$	\\
};

\begin{pgfonlayer}{myback}
\node at (0,1.6) {$U(2)$};
\end{pgfonlayer}
\end{tikzpicture}
\quad
\begin{tikzpicture}
\matrix (m)[matrix of nodes, style={nodes={color = white, rectangle,draw,minimum width=10em, text width = 8.5em, line width=0pt}}, minimum height=2em, row sep=-\pgflinewidth, column sep=-\pgflinewidth,
    column 1/.style={anchor=base west}]
{
\color{black}$\succ(w_{1}): f_{2}, f_{1}, f_{3}$	\\ 
\color{black}$\succ(w_{2}): f_{1}, f_{2}, f_{3}$	\\ 
\color{black}$\succ(w_{3}): f_{2}, f_{1}, f_{3}$	\\
};

\begin{pgfonlayer}{myback}
\node at (0,1.6) {BNE profile};
\end{pgfonlayer}
\end{tikzpicture}
\begin{tikzpicture}
\matrix (m)[matrix of nodes, style={nodes={color = white, rectangle,draw,minimum width=3em, line width=0pt}}, minimum height=2em, row sep=-\pgflinewidth, column sep=-\pgflinewidth]
{
$\lambda_1$	\\ 
\color{black}$\lambda_1$		\\ 
$\lambda_1$	\\
};
\begin{pgfonlayer}{myback}
\fillpattern[north east lines][orange]{m-2-1}{m-2-1}
\end{pgfonlayer}
\end{tikzpicture}
\quad
\begin{tikzpicture}
\matrix (m)[matrix of nodes, style={nodes={rectangle,draw,minimum width=3em}}, minimum height=2em, row sep=-\pgflinewidth, column sep=-\pgflinewidth]
{
$\bm{3,2}$		& $1,5$		 	& $2,2$		\\ 
$2,5$				& $\bm{3,2}$		& $1,5$		\\ 
$1,1$				& $2,1$			& $\bm{3,1}$	\\
};

\begin{pgfonlayer}{myback}
\fillpattern[north east lines][orange]{m-2-1}{m-2-1}
\fillpattern[north east lines][orange]{m-1-2}{m-1-2}
\fillpattern[north east lines][orange]{m-3-3}{m-3-3}
\end{pgfonlayer}
\end{tikzpicture}
\quad
\begin{tikzpicture}
\matrix (m)[matrix of nodes, style={nodes={rectangle,draw,minimum width=3em}}, minimum height=2em, row sep=-\pgflinewidth, column sep=-\pgflinewidth]
{
$\bm{3,2}$		& $1,5$		 	& $2,2$		\\ 
$1,5$				& $\bm{3,2}$		& $2,5$		\\ 
$1,1$				& $2,1$			& $\bm{3,1}$	\\
};

\begin{pgfonlayer}{myback}
\oshade[white][white]{m-1-1}{m-1-1}
\oshade[white][white]{m-1-3}{m-1-3}
\oshade[white][white]{m-2-2}{m-2-2}
\oshade[white][white]{m-2-3}{m-2-3}
\oshade[white][white]{m-3-1}{m-3-1}
\oshade[white][white]{m-3-2}{m-3-2}

\fillpattern[north east lines][orange]{m-3-1}{m-3-1}
\fillpattern[north east lines][orange]{m-1-2}{m-1-2}
\fillpattern[north east lines][orange]{m-2-3}{m-2-3}
\end{pgfonlayer}
\end{tikzpicture}
\quad
\begin{tikzpicture}
\matrix (m)[matrix of nodes, style={nodes={color = white, rectangle,draw,minimum width=10em, text width = 8.5em, line width=0pt}}, minimum height=2em, row sep=-\pgflinewidth, column sep=-\pgflinewidth,
    column 1/.style={anchor=base west}]
{
\color{black}$\succ'(w_{1}): f_{2}, f_{3}$	\\ 
\color{black}$\succ'(w_{2}): f_{1}, f_{2}, f_{3}$   \\ 
\color{black}$\succ'(w_{3}): f_{2}, f_{3}$	\\
};
\end{tikzpicture}\\
\begin{tikzpicture}
\matrix (m)[matrix of nodes, style={nodes={color = white, rectangle,draw,minimum width=3em, line width=0pt}}, minimum height=2em, row sep=-\pgflinewidth, column sep=-\pgflinewidth]
{
$\lambda_2$	\\ 
\color{black}$\lambda_2$		\\ 
$\lambda_2$	\\
};
\begin{pgfonlayer}{myback}
\fillpattern[north east lines][lightblue]{m-2-1}{m-2-1}
\end{pgfonlayer}
\end{tikzpicture}
\quad
\begin{tikzpicture}
\matrix (m)[matrix of nodes, style={nodes={rectangle,draw,minimum width=3em}}, minimum height=2em, row sep=-\pgflinewidth, column sep=-\pgflinewidth]
{
$\bm{3,2}$		& $1,5$		 	& $2,2$		\\ 
$2,5$				& $\bm{3,2}$		& $1,5$		\\ 
$1,1$				& $2,1$			& $\bm{3,1}$	\\
};

\begin{pgfonlayer}{myback}
\fillpattern[north east lines][lightblue]{m-2-1}{m-2-1}
\fillpattern[north east lines][lightblue]{m-3-2}{m-3-2}
\fillpattern[north east lines][lightblue]{m-1-3}{m-1-3}
\end{pgfonlayer}
\end{tikzpicture}
\quad
\begin{tikzpicture}
\matrix (m)[matrix of nodes, style={nodes={rectangle,draw,minimum width=3em}}, minimum height=2em, row sep=-\pgflinewidth, column sep=-\pgflinewidth]
{
$\bm{3,2}$		& $1,5$		 	& $2,2$		\\ 
$1,5$				& $\bm{3,2}$		& $2,5$		\\ 
$1,1$				& $2,1$			& $\bm{3,1}$	\\
};

\begin{pgfonlayer}{myback}
\oshade[white][white]{m-1-1}{m-1-1}
\oshade[white][white]{m-1-3}{m-1-3}
\oshade[white][white]{m-2-2}{m-2-2}
\oshade[white][white]{m-2-3}{m-2-3}
\oshade[white][white]{m-3-1}{m-3-1}
\oshade[white][white]{m-3-2}{m-3-2}

\fillpattern[north east lines][lightblue]{m-3-1}{m-3-1}
\fillpattern[north east lines][lightblue]{m-1-2}{m-1-2}
\fillpattern[north east lines][lightblue]{m-2-3}{m-2-3}
\end{pgfonlayer}
\end{tikzpicture}
\quad
\begin{tikzpicture}
\matrix (m)[matrix of nodes, style={nodes={color = white, rectangle,draw,minimum width=10em, text width = 8.5em, line width=0pt}}, minimum height=2em, row sep=-\pgflinewidth, column sep=-\pgflinewidth,
    column 1/.style={anchor=base west}]
{
\color{black}$\succ''(w_{1}): f_{2}, f_{3}$	\\ 
\color{black}$\succ''(w_{2}): f_{1}, f_{3}, f_{2}$   \\ 
\color{black}$\succ''(w_{3}): f_{2}, f_{1}, f_{3}$	\\
};
\end{tikzpicture}
\\
\begin{tikzpicture}
\matrix (m)[matrix of nodes, style={nodes={color = white, rectangle,draw,minimum width=3em, line width=0pt}}, minimum height=2em, row sep=-\pgflinewidth, column sep=-\pgflinewidth]
{
$\lambda_3$	\\ 
\color{black}$\lambda_3$		\\ 
$\lambda_3$	\\
};
\begin{pgfonlayer}{myback}
\fillpattern[north east lines][pear]{m-2-1}{m-2-1}
\end{pgfonlayer}
\end{tikzpicture}
\quad
\begin{tikzpicture}
\matrix (m)[matrix of nodes, style={nodes={rectangle,draw,minimum width=3em}}, minimum height=2em, row sep=-\pgflinewidth, column sep=-\pgflinewidth]
{
$\bm{3,2}$		& $1,5$		 	& $2,2$		\\ 
$2,5$				& $\bm{3,2}$		& $1,5$		\\ 
$1,1$				& $2,1$			& $\bm{3,1}$	\\
};

\begin{pgfonlayer}{myback}
\fillpattern[north east lines][pear]{m-2-1}{m-2-1}
\fillpattern[north east lines][pear]{m-1-2}{m-1-2}
\end{pgfonlayer}
\end{tikzpicture}
\quad
\begin{tikzpicture}
\matrix (m)[matrix of nodes, style={nodes={rectangle,draw,minimum width=3em}}, minimum height=2em, row sep=-\pgflinewidth, column sep=-\pgflinewidth]
{
$\bm{3,2}$		& $1,5$		 	& $2,2$		\\ 
$1,5$				& $\bm{3,2}$		& $2,5$		\\ 
$1,1$				& $2,1$			& $\bm{3,1}$	\\
};

\begin{pgfonlayer}{myback}
\oshade[white][white]{m-1-1}{m-1-1}
\oshade[white][white]{m-1-3}{m-1-3}
\oshade[white][white]{m-2-2}{m-2-2}
\oshade[white][white]{m-2-3}{m-2-3}
\oshade[white][white]{m-3-1}{m-3-1}
\oshade[white][white]{m-3-2}{m-3-2}

\fillpattern[north east lines][pear]{m-3-1}{m-3-1}
\fillpattern[north east lines][pear]{m-1-2}{m-1-2}
\fillpattern[north east lines][pear]{m-2-3}{m-2-3}
\end{pgfonlayer}
\end{tikzpicture}
\quad
\begin{tikzpicture}
\matrix (m)[matrix of nodes, style={nodes={color = white, rectangle,draw,minimum width=10em, text width = 8.5em, line width=0pt}}, minimum height=2em, row sep=-\pgflinewidth, column sep=-\pgflinewidth,
    column 1/.style={anchor=base west}]
{
\color{black}$\succ'''(w_{1}): f_{2}, f_{3}, f_{1}$	\\ 
\color{black}$\succ'''(w_{2}): f_{1}, f_{2}, f_{3}$   \\ 
\color{black}$\succ'''(w_{3}): f_{2}$	\\
};
\end{tikzpicture}

\end{center}